\documentclass[copyright]{eptcs}
\providecommand{\event}{J.~Dassow, G.~Pighizzini, B.~Truthe (Eds.): 11th International Workshop\\
on Descriptional Complexity of Formal Systems (DCFS 2009)} 
\providecommand{\volume}{3}
\providecommand{\anno}{2009}
\providecommand{\firstpage}{157} 
\providecommand{\eid}{15}
\usepackage{breakurl}        

\input{dcfs.tex}

\usepackage{bm}
\usepackage{verbatim}
\usepackage{graphicx}

\newcommand{\makeset}[2]{\ensuremath{ \{ #1 \: | \: #2 \} }}

\renewcommand{\emptyset}{\varnothing}
\renewcommand{\epsilon}{\varepsilon}

\newcommand{\D}[1]{\texttt{\textup{#1}}} 
\newcommand{\basek}[2][k]{\bm{(}#2\bm{)}_{#1}} 
\newcommand{\plus}{\mathop{\boxplus^R}} 

\begin{document}

\title{Nondeterministic State Complexity \\ of Positional Addition}
\def\titlerunning{Nondeterministic State Complexity of Positional Addition}
\def\authorrunning{G.~Jir\'askov\'a, A.~Okhotin}

\author{Galina Jir\'askov\'a\,\thanks{Research supported by VEGA grant 2/0111/09.}
\institute{Mathematical Institute -- Slovak Academy of Sciences\\
	Ko\v{s}ice -- Slovakia}
\email{jiraskov@saske.sk}
\and
Alexander Okhotin\,\thanks{Research supported by the Academy of Finland under grant 118540.}
\institute{Academy of Finland}
\institute{Department of Mathematics -- University of Turku\\
        Turku -- Finland}
\email{alexander.okhotin@utu.fi}
}

\maketitle

\begin{abstract}
Consider nondeterministic finite automata
recognizing base-$k$ positional notation of numbers.
Assume that numbers are read
starting from their least significant digits.
It is proved that if two sets of numbers $S$ and $T$
are represented by 
nondeterministic automata of $m$ and $n$ states, respectively,
then their sum $\{s+t \mid s \in S, \: t \in T\}$
is represented by a nondeterministic  automaton with $2mn+2m+2n+1$ states.
Moreover, this number of states is necessary in the worst case
for all~$k \geqslant 9$. 
\end{abstract}

\section{Introduction}


Descriptional complexity of operations on regular languages
with respect to their representation by finite automata and regular expressions
is among the common topics of automata theory.
With respect to deterministic finite automata (DFAs),
and using the number of states as a complexity measure,
the \emph{state complexity} of 
basic operations on languages
was determined by Maslov~\cite{Maslov} in 1970.
In particular, such results as
``if languages $K$ and $L$ are recognized by 
DFAs of $m$ and $n$ states, respectively,
then the language $KL$ requires a DFA with up to $(2m-1)2^{n-1}$ states''
originate from that paper.

Over the last two decades,
similar results were obtained for nondeterministic finite automata (NFAs).
In particular, Birget \cite{Birget1993} has shown that
the complement of a language recognized by an $n$-state NFA
may require an NFA with as many as $2^n$ states,
and this result was later improved by Jir\'askov\'a~\cite{ji05}
who reduced the alphabet for the witness language
from $\{a,b,c,d\}$ to $\{a,b\}$.
The systematic study of
\emph{nondeterministic state complexity},
that is, state complexity with respect to NFAs,
of different operations
was started by Holzer and Kutrib~\cite{HolzerKutrib},
who obtained, in particular,
the precise results for union, intersection and concatenation.	
More recently Jir\'askov\'a and Okhotin~\cite{cyclic_shift}
determined the nondeterministic state complexity of cyclic shift,
Gruber and Holzer~\cite{GruberHolzerKutribMCU}
established precise results for scattered substrings
and scattered superstrings,
Domaratzki and Okhotin~\cite{power_sc}
studied $k$-th power of a language, $L^k$,
while Han, K. Salomaa and Wood~\cite{hsw09}
considered the standard operations on NFAs
in the context of prefix-free languages.

The present paper continues this study
by investigating another operation,
which has recently been used
by Je\.z and Okhotin~\cite{JezOkhotin_CSR,JezOkhotin_completeness}
in the study of language equations.
This is the operation of addition of strings in base-$k$ positional notation.
Let $\Sigma_k=\{\D0, \D1, \ldots, k-1\}$ with $k \geqslant 2$
be an alphabet of digits.
Then a string $a_{\ell-1} \cdots a_0 \in \Sigma_k^*$
represents a number
$\basek{a_{\ell-1} \cdots a_0}=\sum_{i=0}^{\ell-1} a_i \cdot k^i$,
and there is a correspondence between natural numbers
and strings in $\Sigma_k^* \setminus \D0\Sigma_k^*$.
For two strings $u, v \in \Sigma_k^* \setminus \D0\Sigma_k^*$,
their sum can be defined as $w=u \boxplus v$ 
as the unique string $w \in \Sigma_k^* \setminus \D0\Sigma_k^*$,
for which $\basek{w}=\basek{u}+\basek{v}$.
The operation extends to languages as follows:
for all $K, L \subseteq \Sigma_k^* \setminus \D0\Sigma_k^*$,
$K \boxplus L
		=
	\makeset{u \boxplus v}{u \in K, \: v \in L}$.


This operation preserves regularity,	
and proving that
can be regarded as an exercise in automata theory. 
The paper begins with a solution to this exercise,
given in Section~\ref{section_upper_bound}.
For convenience, it is assumed
that automata read a notation of a number
starting from its least significant digit;
to put it formally,
a slightly different operation is studied:
	$K \plus L 
		=
	(K^R \boxplus L^R)^R$.
This variant seems to be
more natural in the context of automata,
and furthermore,
since the nondeterministic state complexity of reversal is $n+1$,
the complexity of these two operations is almost the same.

The straightforward construction of an automaton
recognizing the language $L(A) \plus L(B)$
for an $m$-state NFA $A$
and an $n$-state NFA $B$
yields an NFA with $2mn+2m+2n+1$ states.
The purpose of this paper is to show
that this construction is in fact optimal,
and there are witness languages,
for which exactly this number of states is required.
This is established in Section~\ref{section_lower_bounds},
where worst-case automata are presented
for $m, n \geqslant 1$ with $m+n \geqslant 3$.
The case of $m=n=1$ requires a special treatment,
and it is proved that the NFA recognizing a positional sum
of two one-state automata requires 6 states in the worst case.

\section[Constructing an NFA]{Constructing an NFA for $K \plus L$} \label{section_upper_bound}

A \emph{nondeterministic finite automaton} (NFA)
is a quintuple $A=(Q, \Sigma, \delta, q_0, F)$,
in which $Q$ is a finite set of states,
$\Sigma$ is a finite input alphabet,
$\delta : Q \times \Sigma \to 2^Q$
is the (nondeterministic) transition function,
$q_0 \in Q$ is the initial state, and
$F \subseteq Q$ is the set of accepting states.
An NFA is called a \emph{deterministic finite automaton} (DFA)
if $|\delta(q, a)|=1$ for all $q$ and $a$,
and it is a \emph{partial DFA}
if $|\delta(q, a)| \leqslant 1$.
The transition function can be naturally extended to the domain $Q \times \Sigma^*$.
The language recognized by the NFA $A$,
denoted  $L(A)$,
 is the set $\{w \in \Sigma^* \mid \delta(q_0,w) \cap F \neq \emptyset \}$.

Throughout this paper, the letters in an alphabet of size $k$
are always considered as digits in base-$k$ notation,
and the alphabet is 
$\Sigma_k=\{\D0, \D1, \ldots, k-1\}$.
With such an alphabet fixed,
the \emph{nondeterministic state complexity 
of positional addition} of NFAs is defined as a function
$f_k \colon \mathbb{N} \times \mathbb{N} \to \mathbb{N}$,
where $f_k(m, n)$ is the least number of states in an NFA
sufficient to represent $L(A) \plus L(B)$
for every $m$-state NFA $A$ and $n$-state NFA $B$
with $L(A), L(B) \subseteq \Sigma_k^* \setminus \D0 \Sigma_k^*$.
The following lemma,
besides formally establishing
that regular languages are closed
under addition in positional notation,
gives an upper bound on this function.

\begin{lemma}\label{le:upper_bound}
Let $A$ and $B$ be NFAs
over 
$\Sigma_k=\{\D0, \D1, \ldots, k-1\}$
with $m$ and $n$ states, respectively.
Let $L(A) \cap \D0 \Sigma_k^* = L(B) \cap \D0 \Sigma_k^* = \emptyset$.
Then there exists 
a $(2mn+2m+2n+1)$-state NFA over $\Sigma_k$ for
the language
$L(A) \plus L(B)$.
\end{lemma}
\begin{proof}
Let $A=(P, \Sigma_k, \delta_A, p_0, F_A)$
and $B=(Q, \Sigma_k, \delta_B, q_0, F_B)$.
The new NFA $C$ has a set of states split into four groups: 
 $\widehat{Q} = Q^{AB} \cup Q^A \cup Q^B \cup \{q_{acc}\}$, where
\begin{align*}
	Q^{AB} &= P \times Q \times \{0, 1\}, \\
	Q^A &= \{A\} \times P \times \{0, 1\}, \\
	Q^B &= \{B\} \times Q \times \{0, 1\}.
\end{align*}

(I) Each state $(p, q, c) \in Q^{AB}$
corresponds to $A$ in state $p$, $B$ in state $q$ and carry digit $c \in \{0, 1\}$.
In particular, the state $(p_0, q_0, 0)$ is the initial state of this NFA.
State $(p, q, c)$ represents the case
shown in Figure~\ref{f:nfa_for_plus_addition}(left).
A string of digits $ddddd$ has been read,
and $C$ has guessed its representation as a sum of two strings of digits,
$aaaaa \plus bbbbb$,
where $A$ goes to $p$ by $aaaaa$
and $B$ goes to $q$ by $bbbbb$.
If $c=1$, then $aaaaa \plus bbbbb = \D1ddddd$.

The transitions from one state of this kind to another
are defined as follows.
Suppose $A$ reads a digit $a$ 
and goes from $p$ to $p'$,
while $B$ may go from $q$ to $q'$ by a digit $b$. 
Then, taking the carry digit~$c$ into account,
the sum may contain a digit $a+b+c$ or $a+b+c-k$ in this position
depending on whether $a+b+c < k$ or not,
and also the carry should be adjusted accordingly.
Thus $C$ has a transition
from $(p, q, c)$ to
to $(p',q',0)$ by $a+b+c$ if $a+b+c < k$,
or a transition to $(p',q',1)$ by $a+b+c-k$ if $a+b+c \geqslant k$.
This procedure continues
until the string of digits recognized by $A$ or by $B$ finishes.
Then $C$ enters a state of one of the following two groups.
\begin{figure}[hbt]
	\centerline{\includegraphics[scale=0.83]{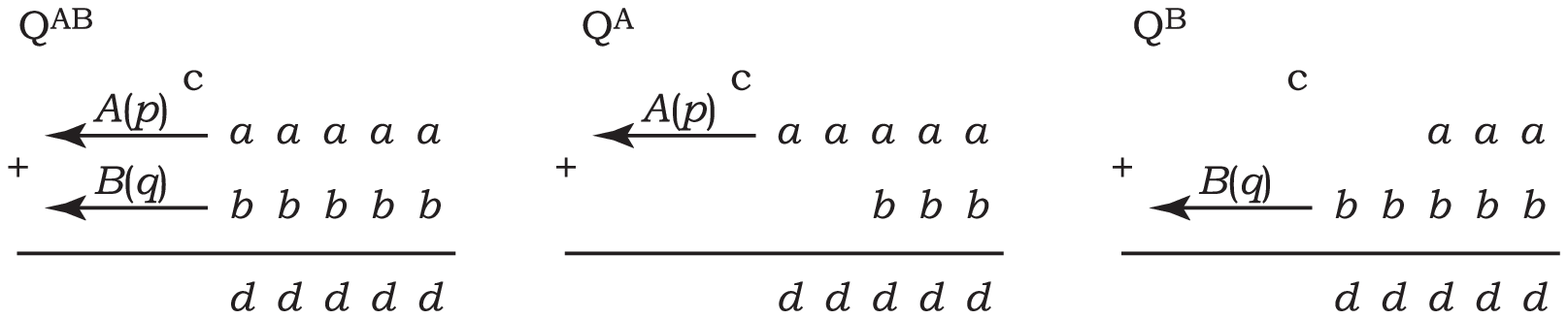}}
	\caption{Transitions out of $(p,q,0) \in Q^{AB}$ in the constructed NFA.}
	\label{f:nfa_for_plus_addition}
\end{figure}

(II) If the automaton $B$ is no longer running
(that is, the notation of the second number has ended),
while $A$ still produces some digits,
this case is implemented in states $(A, p, c) \in Q^A$,
where $p$ is a state of~$A$ and~$c$ is a carry.
This case is illustrated in Figure~\ref{f:nfa_for_plus_addition}(middle).
The NFA $C$ reaches this group of states as follows.
For every state $(p, q, c) \in Q^{AB}$,
such that $q$ 
is an accepting state of $B$,
the string recognized by $B$ can be pronounced finished.
Suppose that $A$ may go from $p$ to $p'$
by a digit $a$. 
Then the sum may contain a digit $a+c$ or $a+c-k$.
This case is represented by a transition of $C$
from $(p, q, c)$
 to $(A,p',0)$ by $a+c$ if $a+c<k$,
 or to $(A,p',1)$ by $a+c-k$ if $a+c\geqslant k$.
Once $C$ enters the subset $Q^A$,
it can continue reading the number as follows.
For every state $(A, p, c)$,
if $A$ may go from $p$ to $p'$ by a digit $a$,
then there is a transition from $(A, p, c)$
to $(A,p',0)$ by $(a+c)$ if $a+c< k$,
or to $(A,p',1)$ by $(a+c-k)$ if $a+c \geqslant k$.

(III) Symmetrically, there is a group of states $(B, q, c)$,
which correspond to the case
when the number read by $A$ has ended.
For each state $(p, q, c) \in Q^{AB}$ with $p \in F_A$,
for every digit $b$ 
and for every state~$q'$,
such that $B$ has a transition from $q$ to $q'$ by $b$,
the new automaton $C$ has a transition
from $(p, q, c)$
to $(B,q',0)$ by $b+c$ if $b+c<k$,
or to $(B,q',1)$ by $b+c-k$ if $b+c \geqslant k$.
Second, for every state $(B, q, c)$,
if $B$ may go from $q$ to $q'$ by a digit $b$,
then $C$ has a transition from $(B, q, c)$
to $(B,q',0)$ by $(b+c)$ if $b+c< k$,
or to $(B,q',1)$ by $(b+c-k)$ if $b+c \geqslant k$.

(IV) $q_{acc}$ is a special accepting state with no outgoing transitions.
This state is needed when the strings of digits
recognized by $A$ and $B$ have already finished, but the carry digit remains,
and thus an extra input symbol has to be read.
The automaton $C$ reaches this state by reading the digit $\D1$
under the following conditions:
for all $p \in F_A$ and $q \in F_B$,
there are transitions by $\D1$
from $(p, q, 1)$,
from $(A, p, 1)$
and from $(B, q, 1)$
to $q_{acc}$.

The other accepting states are all states of the form
$(p, q, 0)$, $(A, p, 0)$ and $(B, q, 0)$,
with $p \in F_A$ and~$q \in F_B$.

This completes the construction.
The general form of transitions
from a state $(p, q, c) \in Q^{AB}$
is illustrated in Figure~\ref{f:nfa_for_plus},
separately for $c=0$ and $c=1$.
\end{proof}
\begin{figure}[hbt]
	\centerline{\includegraphics[scale=0.83]{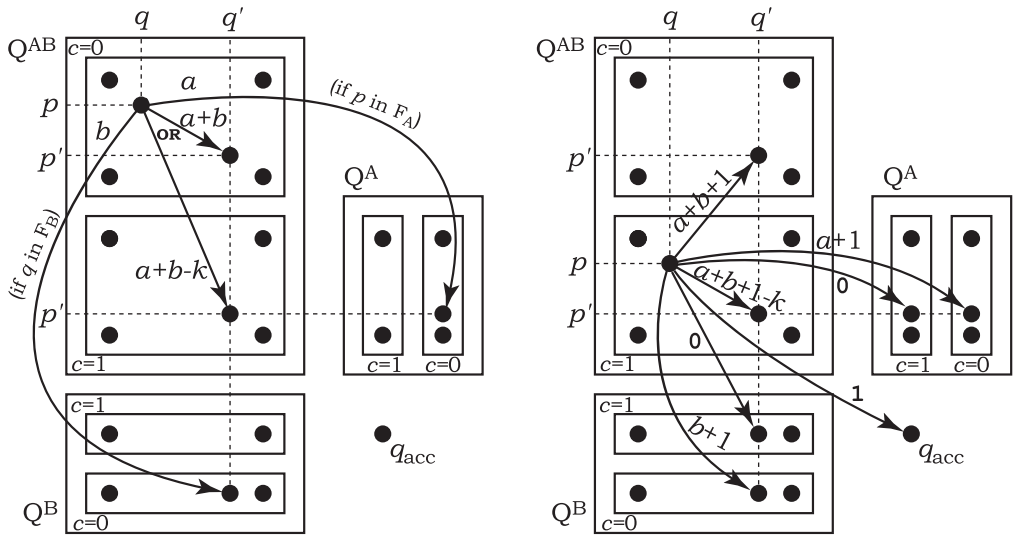}}
	\caption{Transitions out of $(p, q, 0)$
	and out of $(p, q, 1)$ in the constructed NFA.}
	\label{f:nfa_for_plus}
\end{figure}

\section{Lower bounds} \label{section_lower_bounds}

The goal of the paper is to prove
that the $2mn+2m+2n+1$ bound of Lemma~\ref{le:upper_bound} is tight.
As this requires a rather difficult  
proof,
the following weaker result will be established first.

\begin{lemma} 
Let $\Sigma_k=\{\D0, \D1, \ldots, k-1\}$
be an alphabet with $k \geqslant 2$.
Let $m,n \geqslant 1$ be relatively prime numbers
and consider languages $L_m=(\D1^m)^*$ and $L_n=(\D1^n)^*$,
which are representable by NFAs of
$m$ and $n$ states, respectively.
Then every NFA recognizing the language $L_m \plus L_n$
has at least $mn$ states.
\end{lemma}
\begin{proof}
Let $A$ be an NFA for $L_m \plus L_n$
with $\ell$ states.
If $k \geqslant 3$,
construct a new $\ell$-state NFA $B$
recognizing $(L_m \plus L_n) \cap \D2^*$
which can be done by taking the NFA $A$
and omitting transitions by all symbols except for $\D2$.
Then $L(B)=(\D2^{mn})^*$.
This is a language that requires an NFA of at least $mn$ states.
Therefore, \hbox{$\ell \geqslant mn$}. 
In the case of $k=2$,
let $B$ recognize $(L_m \plus L_n) \cap \D0\D1^*$.
In this case it is sufficient to have $\ell+1$ states in $B$,
and $L(B)=\D0 (\D1^{mn})^*$.
As this language requires
an NFA with at least $mn+1$ states,
the statement is proved.
\end{proof}

In order to prove a precise lower bound,
a different
construction of witness languages is needed.
At present, the witness languages
are defined over an alphabet of at least nine symbols,
that is,
the bound applies to addition in base 9 or greater.
Lower bounds on the resulting languages of sums
will be proved using the well-known
fooling-set lower bound technique.
After defining a fooling set
we recall the lemma 
describing the technique, 
and give a small example.
Then, the lower bound result follows.

\begin{definition}
\label{def_fool}
A set of pairs of strings $ \{ (x_i,y_i)\mid i=1,2,\ldots,n\} $ 
is said to be a~\textit{fooling set} for a  language~$L$ 
if for every $ i $ and $ j $ in $ \{ 1,2,\ldots,n\}, $
\begin{description}
\item[(F1)]
	the string $x_{i}y_{i}$ is in the language $ L,$
\item[(F2)]
	if $i\neq j,$ then at least one of the strings $x_{i}y_{j}$ 
and $x_{j}y_{i}$ is not in  $ L.$
\end{description}
\end{definition}

\begin{lemma}[{(Birget~\cite{Birget})}]
\label{fool}
Let 
$\mathcal{A}$
be a~fooling set for a regular language $L$.
Then every NFA recognizing the language $L$ requires at least $|\mathcal{A}|$ states.
\end{lemma}

\begin{example}
\label{ex} 
Consider the regular language  
$ L = \{  w \in \Sigma^* \mid$
	the number of $a$'s in $w$	is a multiple of $n \}$.
The set of pairs of strings
$ \{  (a,a^{n-1}),(a^2,a^{n-2}), \ldots, (a^n,\varepsilon)  \} $ 
is a fooling set for the language $ L $ because for every $ i $ and $ j $ in $ \{1,2,\ldots,n\} ,$
\begin{description}
\item[(F1)]
	$ a^ia^{n-i}=a^n$, and the string $a^n$ is in the language $L$, and 
\item[(F2)]
	if $ 1 \le i < j \le n, $ then $ a^ia^{n-j}=a^{n-(j-i)}$, 
	and the string $ a^{n-(j-i)}$ is not in the language $L$	
	since $ 0< n-(j-i) < n$.
\end{description}
Hence by Lemma~\ref{fool}, every NFA for the language $L$ needs at least $n$ states.%
\hfill$\diamond$
\end{example}

\begin{lemma}\label{nfa_lower_bound_lemma}
Let $\Sigma_k=\{\D0, \D1, \ldots, k-1\}$ be an alphabet with $k \geqslant 9$.
Let $m \geqslant 1$ and $n \geqslant 2$,
and consider the partial DFAs $A_m$ and $B_n$ over $\Sigma_k$
given in Figure~\ref{f:nfa_lower_bound}.
Then 
every  NFA
for  $L(A_m) \plus L(B_n)$
has at least $2mn+2m+2n+1$ states.
\end{lemma}

\begin{proof}
In plain words,  $L(A_m)$ represents
all numbers with their base-$k$ notation
using only digits $\D1$, $\D2$ and $k-1$,
with the number of $\D1$s equal to $m-1$ modulo $m$.
Similarly, the base-$k$ notation
of all numbers in  $L(B_n)$
uses only digits $\D1$, $\D3$ and $k-1$,
and the total number of $\D1$s and $(k-1)$s
should be $n-1$ modulo~$n$.

\begin{figure}[hbt]
	\centerline{\includegraphics[scale=0.83]{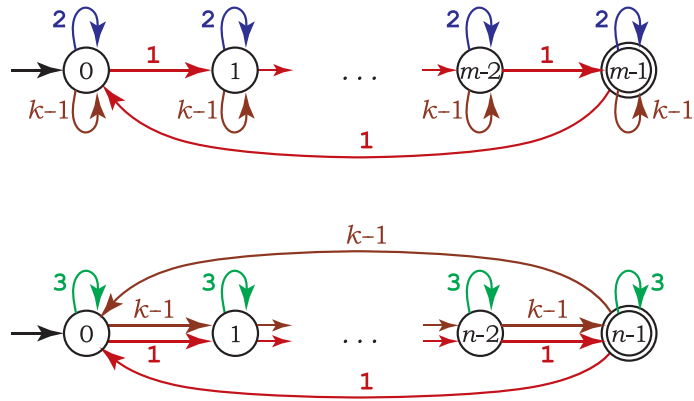}}
	\caption{The nondeterministic finite automata $A_m$ and $B_n$
		over $\Sigma_k=\{\D0, \D1, \ldots, k-1\}$ with $k \geqslant 9$.}
	\label{f:nfa_lower_bound}
\end{figure}

Let the set of states of $A_m$ be $P=\{0, \ldots, m-1\}$
and let the states of $B_n$ be $Q=\{0, \ldots, n-1\}$.
Let $L=L(A_m) \plus L(B_n)$, and let us
construct a $(2mn+2m+2n+1)$-state NFA 
$$M=(Q^{AB} \cup Q^A \cup Q^B \cup \{q_{acc}\},\Sigma_k,\delta,q_0,F)$$
for the language $L$ 
as in Lemma~\ref{le:upper_bound}.
The initial state of $M$ is $q_0=(0, 0, 0)$.
The full set of transitions
is omitted due to space constraints;
the reader can reconstruct it
according to Lemma~\ref{le:upper_bound}.
The below incomplete list
represents all information about $M$
used later in the proof:
\begin{itemize} 
 \item Each state $(i,j,0)$ goes 
	to itself by \D5;
	to  state $(i,j+1,0)$ by \D3; 
	to state $(i+1,j,0)$ by $4$,  and
	to state $(i,j+1,1)$ by $k-2$.
	Each state $(m-1,j,0)$ also goes
	 to state $(B,j,0)$ by \D3.
	 

  \item Each state $(i,j,1)$ goes 
	 to state $(i,j,0)$ by \D6. 
	%
	Each state $(i,n-1,1)$ also goes 
		to state $(A,i,1)$ by~$\D0$,  
	 %
	and each state  $(m-1,j,1)$ also goes
		to state $(B,j+1,1)$  by $\D0$. 
	 %
 \item Each state $(A,i,1)$ goes 
	to itself by $\D0$;
	to state $(A,i,0)$ by $\D3$; and
	to state $(A,i+1,0)$ by $\D2$.
	
 \item Each state $(A,i,0)$  goes 
	to itself by $\D2$ and  $k-1$; and
	to  $(A,i+1,0)$ by $\D1$.
	
 \item Each state $(B,j,1)$ goes	
	to state $(B,j+1,1)$ by $\D0$; 
	to state $(B,j,0)$ by $\D4$; and
	to state $(B,j+1,0)$ by $\D2$.
 
 \item Each state $(B,j,0)$ goes	 
	to itself by $\D3$; and
	to  $(B,j+1,0)$ by $\D1$ and  $k-1$. 
 \item State $(A,m-1,1)$ goes to state $q_{acc}$ by $\D1$.	
\end{itemize}
Notice that in states  
$(A,i,c)$ and $(B,j,c)$, 
transitions by $\D5$ and by $\D6$  
are not defined, and no transitions are defined in state $q_{acc}$. 
There are four accepting states:
$(m-n, n-1, 0)$, $(A, m-1, 0)$, $(B, n-1, 0)$ and $q_{acc}$.
Transitions from $(i,j,0)$ and $(i,j,1)$
are illustrated in Figure~\ref{f:nfa_lower_bound_M},
where transitions not used in the proof are shown in grey.

\begin{figure}[hbt]
	\centerline{\includegraphics[scale=0.83]{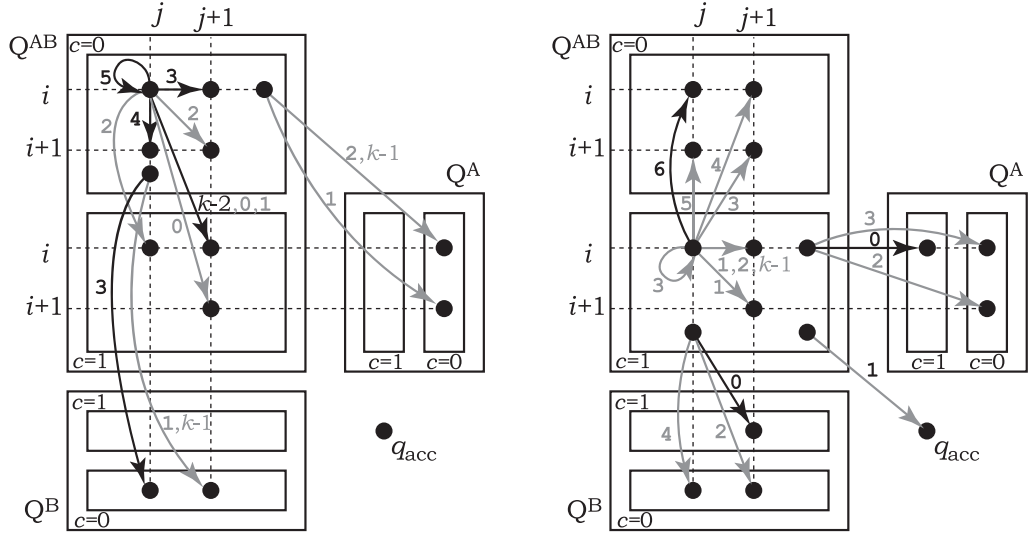}}
	\caption{NFA $M$: transitions out of states $(i,j,0)$ and $(i,j,1)$.}
	\label{f:nfa_lower_bound_M}
\end{figure}

Our goal is to show that every NFA 
for the language $L$ requires at least $2mn+2m+2n+1$ states.
We prove this by describing a fooling set 
for the language $L$ of size $2mn+2m+2n+1$.
Consider the following sets of pairs of strings,
in which the difference $j-1$ is modulo $n$
(that is, $j-1=n-1$ for $j=0$):
\begin{align*}
\mathcal{A} =\:
	&\{(\D4^i\D3^j, \D5\D4^{m-1-i}\D3^{n-1-j}\D5) 
	\mid i=0,1,\ldots,m-1, j=0,1,\ldots,n-1\}, \\
\mathcal{B} =\:
	&\{(\D4^i\D3^{j-1}(k-2), \D6\D4^{m-1-i}\D3^{n-1-j}\D5) 
	\mid i=0,1,\ldots,m-1, j=0,1,\ldots,n-1\}, \\
\mathcal{C} =\:
	&\{(\D4^i\D3^{n-2}(k-2)\D0, \D3\D1^{m-1-i}\D2\D2) 
	\mid i=0,1,\ldots,m-1\} \: \cup \\
	&\{(\D4^i\D3^{n-2}(k-2)\D0\D3,  \D1^{m-1-i}\D2\D2) 
	\mid i=0,1,\ldots,m-1 \}, \\
\mathcal{D} =\:
	&\{(\D4^{m-1}\D3^{n-1}(k-2)\D0\D0^j, \D0^{n-1-j}\D4\D1^{n-1}\D3\D3) 
	\mid j=0,1,\ldots,n-1\} \: \cup \\
	&\{(\D4^{m-1}\D3^{n-1}(k-2)\D0^n\D4\D1^j,  \D1^{n-1-j}\D3\D3) 
	\mid j=0,1,\ldots,n-1 \}.
\end{align*}
Let $\mathcal{F}=\mathcal{A} \cup \mathcal{B} \cup \mathcal{C} \cup \mathcal{D}$.
Let us show that the set $\mathcal{F}$ is a fooling set for 
$L$, that is,
\begin{description}
\item[(F1)]
	for each pair $(x,y)$ in $\mathcal{F}$, the string $xy$ is in $L$;

\item[(F2)]
	if $(x,y)$ and $(u,v)$ are two different pairs in $\mathcal{F}$,
	then $xv \notin L$ or $uy \notin L$.
\end{description}
%
We prove the statement (F1)  by examination of each pair:
\begin{itemize}
\item
	If $(x,y)$ is a pair in $\mathcal{A}$,
	then $xy=\D4^i\D3^j\D5\D4^{m-1-i}\D3^{n-1-j}\D5$.
	The initial state $(0,0,0)$ of  $M$ goes to state $(i,j,0)$ by $\D4^i\D3^j$, 
	which goes to itself by \D5, and then to the accepting state $(m-1,n-1,0)$ 
	by $\D4^{m-1-i}\D3^{n-1-j}\D5$. 
	Thus 
	$xy$ is accepted by 
	$M$, 
	and so is in  $L$. This case is illustrated in Figure~\ref{pair_in_AB}, left.

\item If $(x,y)$ is a pair in $\mathcal{B}$, 
	then $xy=\D4^i\D3^{j-1}(k-2)\D6\D4^{m-1-i}\D3^{n-1-j}\D5$.
	State $(0,0,0)$ goes to state $(i,j-1,0)$ by $\D4^i\D3^{j-1}$, 
	which goes to state $(i,j,1)$ by $k-2$. 
	State $(i,j,1)$ goes to state $(i,j,0)$ by \D6, 
	and then to the accepting state $(m-1,n-1,0)$ 
	by $\D4^{m-1-i}\D3^{n-1-j}\D5$,
	which is shown in Figure~\ref{pair_in_AB}, right.%
\begin{figure}[hbt]
	\centerline{\includegraphics[scale=0.35]{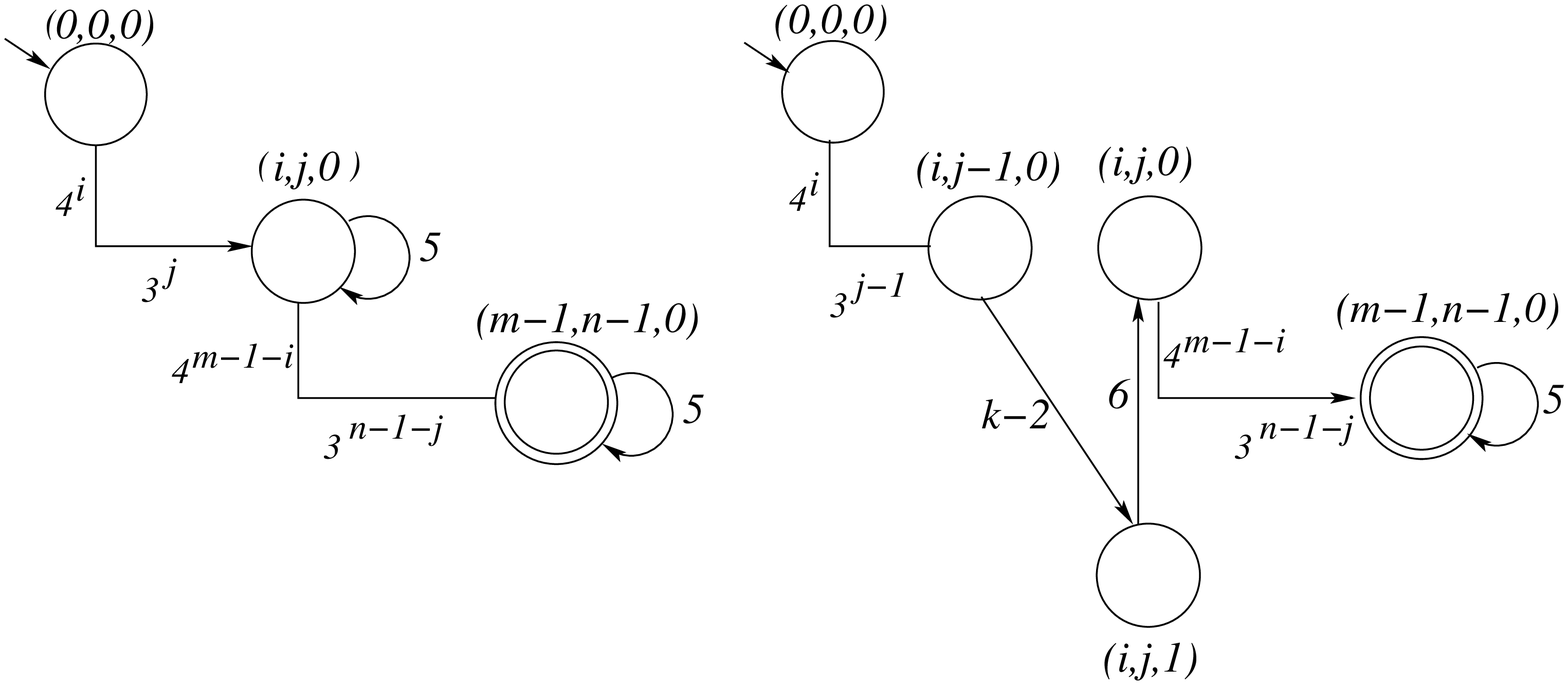}}
	\caption{A pair in $\mathcal{A}$ and a pair in $\mathcal{B}$.}
	\label{pair_in_AB}
\end{figure}
\item If $(x,y)$ is a pair in $\mathcal{C}$, 
	then $xy=\D4^i\D3^{n-2}(k-2)\D0\D3\D1^{m-1-i}\D2\D2$.
	State $(0,0,0)$ goes to state $(i,n-1,1)$ by $\D4^i\D3^{n-2}(k-2)$, 
	which goes to state $(A,i,1)$ by \D0, and then to state $(A,i,0)$ by $\D3$, 
	and 
	to the accepting state $(A,m-1,0)$ by $\D1^{m-1-i}\D2\D2$. 
	This computation path is presented in Figure~\ref{pair_in_CD}, left.

\item If $(x,y)$ is a pair in $\mathcal{D}$, 
	then $xy=\D4^{m-1}\D3^{n-1}(k-2)\D0^n\D4\D1^{n-1}\D3\D3$.
	State $(0,0,0)$ goes to 
	$(m-1,0,1)$ by $\D4^{m-1}\D3^{n-1}(k-2)$, 
	which goes to state $(B,1,1)$ by \D0, and then to state $(B,0,1)$ by $\D0^{n-1}$, 
	and to state $(B,0,0)$ by \D4,
	and 
	to the accepting state $(B,n-1,0)$ by $\D1^{n-1}\D3\D3$,
	as shown in Figure~\ref{pair_in_CD}, right.
\begin{figure}[hbt]
	\centerline{\includegraphics[scale=0.35]{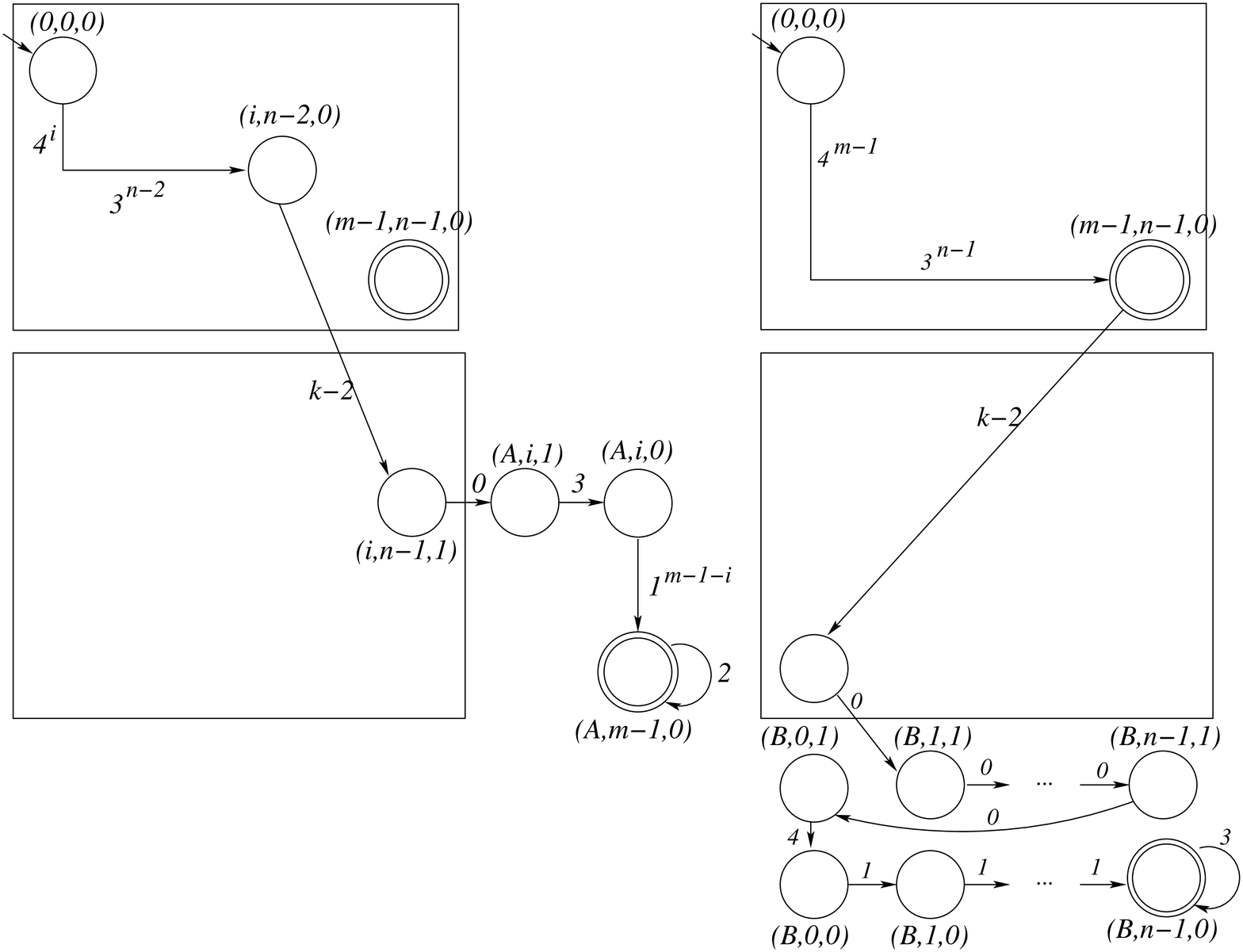}}
	\caption{A pair in $\mathcal{C}$ and a pair in $\mathcal{D}$.}
	\label{pair_in_CD}
\end{figure}
\end{itemize}
Thus in all four cases, the string $xy$ is accepted by the NFA $M$,
and so is in the language $L$.
This proves (F1).
To prove (F2) let us consider the following seven cases:

\begin{itemize}
\item	If $(x,y)$ and $(u,v)$ are two different pairs in $\mathcal{A}$, then
	$$(x,y)=(\D4^i\D3^j,\D5\D4^{m-1-i}\D3^{n-1-j}\D5) \mbox{ and }
	(u,v)=(\D4^{r}\D3^{s},\D5\D4^{m-1-r}\D3^{n-1-s}\D5),$$
	where $(i,j)\neq(r,s)$. 
	Consider the string $xv=\D4^i\D3^j\D5\D4^{m-1-r}\D3^{n-1-s}\D5$. 
	Since the digit \D5 cannot be read in any state $(B,p,0)$,
	after reading $xv$, the NFA $M$ may only be 
	in  state $$(m-1-r+i,n-1-s +j,0).$$ 
	This state is rejecting if $i\neq r$ or $j\neq s$. So the string $xv$ is not in $L$.
\item If $(x,y)$ is a pair in $\mathcal{A}$ and $(u,v)$ is a pair in $\mathcal{B}$,
	then $x=\D4^i\D3^j$  and $v=\D6w$ for a string $w$. 
	After reading $x$, the NFA $M$ is either in state $(i,j,0)$
	or in a state $(B,p,0)$. In these states, transitions by~\D6 are not defined.
	Thus the string $xv$ is rejected by $M$, and so is not in $L$.
\item If  $(x,y)$ is a pair in $\mathcal{A} \cup \mathcal{B}$,
	and $(u,v)$ is a pair in $\mathcal{C} \cup \mathcal{D}$,
	then $y=\D5w$ or $y=\D6w$ for a string $w$. Let us show that the string $uy$ is not in $L$.
	Notice that after reading the string $u$, 
	the NFA $M$ is either in a state 
	$(A,p,c)$ or in a state $(B,q,c)$.
	In these states, no transitions by \D5 and by \D6 are defined.
	Therefore, the string $uy$ is not in $L$.
\item	If $(x,y)$ and $(u,v)$ are two different pairs in $\mathcal{B}$, then 
	$(x,y)=(\D4^i\D3^{j-1}(k-2),\D6\D4^{m-1-i}\D3^{n-1-j}\D5)$ and 
	$(u,v)=(\D4^{r}\D3^{s-1}(k-2),\D6\D4^{m-1-r}\D3^{n-1-s}\D5)$,
	where $(i,j)\neq(r,s)$. After reading $x$,
	the nfa $M$ may only be in state $(i,j,1)$; 
	notice that transitions by $k-2$ are not defined in states $(B,q,0)$.
	State $(i,j,1)$ goes to state $(i,j,0)$ by \D6. 
	From this state, by reading $\D4^{m-1-r}\D3^{n-1-s}\D5$,
	the NFA may only reach the rejecting state $(m-1-r+i,n-1-s +j,0)$. 
	Hence the string $xv$ is not in $L$.
\item	If $(x,y)$ and $(u,v)$ are two different pairs in $\mathcal{C}$, then
	we have three subcases:
	\begin{itemize}
	\item
		$(x,y)=(\D4^i\D3^{n-2}(k-2)\D0, \D3\D1^{m-1-i}\D2\D2)$ and\\
		$(u,v)=(\D4^{r}\D3^{n-2}(k-2)\D0, \D3\D1^{m-1-r}\D2\D2)$, 
		where  $0\leqslant i<r\leqslant m-1$.\\
		After reading $x$, the NFA $M$ is in state $(A,i,1)$, which goes 
		to  state 
		$(A,i,0)$ by $\D3$, and then 
		to rejecting state $(A,m-1-r+i,0)$ 
		by $\D1^{m-1-r}\D2\D2$. Thus $xv$ is not in $L$.
	\item
		$(x,y)=(\D4^i\D3^{n-2}(k-2)\D0\D3, \D1^{m-1-i}\D2\D2)$ and\\
		$(u,v)=(\D4^{r}\D3^{n-2}(k-2)\D0\D3, \D1^{m-1-r}\D2\D2)$, 
		where $0\leqslant i<r\leqslant m-1$.\\
		After reading $x$, the NFA is in state $(A,i,0)$, which goes 
		to rejecting state $(A,m-1-r+i,0)$ 
		by $\D1^{m-1-r}\D2\D2$. Thus $xv$ is not in $L$.
	\item
		$(x,y)=(\D4^i\D3^{n-2}(k-2)\D0, \D3\D1^{m-1-i}\D2\D2)$ and\\
		$(u,v)=(\D4^{r}\D3^{n-2}(k-2)\D0\D3, \D1^{m-1-r}\D2\D2)$.\\
		After reading $u$, the NFA may only be in state $(A,r,0)$,
		where it cannot read symbol 3. Thus $uy$ is not in $L$.
	\end{itemize}
\item	If $(x,y)$ is a pair in $\mathcal{C}$,
	and $(u,v)$ is a pair in $\mathcal{D}$, then
	$y=w\D2\D2$ for a string~$w$. Consider the string~$uy$.
	After reading $u$, the NFA may only be in a state from $Q^B$
	(notice that $n\geqslant 2$). 
	By reading $w$,
	it either hangs,
	or remains in $Q^B$, and then  cannot read $\D2\D2$.
	Therefore, $uy$ is not in $L$.
\item If $(x,y)$ and $(u,v)$ are two different pairs in $\mathcal{D}$,
	then there are three subcases again:
	\begin{itemize}
	\item
		$(x,y)=(\D4^{m-1}\D3^{n-1}(k-2)\D0\D0^j, \D0^{n-1-j}\D4\D1^{n-1}\D3\D3)$ and\\
		$(u,v)=(\D4^{m-1}\D3^{n-1}(k-2)\D0\D0^{s}, \D0^{n-1-s}\D4\D1^{n-1}\D3\D3)$, 
		where \\$0\leqslant j<s\leqslant n-1$.
		Since $n\geqslant 2$, state $(m-1,0,1)$ only goes  to state $(B,1,1)$ by \D0.
		After reading $x$, the NFA is in state $(B,j+1,1)$, which goes 
		to rejecting state $(B,n-1-s+j,0)$ by $\D0^{n-1-s}\D4\D1^{n-1}\D3\D3$.
		Thus $xv$ is not in $L$.
	\item
		$(x,y)=(\D4^{m-1}\D3^{n-1}(k-2)\D0^n\D4\D1^j, \D1^{n-1-j}\D3\D3)$ and\\
		$(u,v)=(\D4^{m-1}\D3^{n-1}(k-2)\D0^n\D4\D1^{s}, \D1^{n-1-s}\D3\D3)$,  
		where $0\leqslant j<s\leqslant n-1$.
		After reading $x$, the NFA is in state $(B,j,0)$, which goes 
		to rejecting state $(B,n-1-s+j,0)$ 
		by $\D1^{n-1-s}\D3\D3$. Thus $xv$ is not in $L$.
	\item
		$(x,y)=(\D4^{m-1}\D3^{n-1}(k-2)\D0\D0^j, \D0^{n-1-j}\D4\D1^{n-1}\D3\D3)$ and\\
		$(u,v)=(\D4^{m-1}\D3^{n-1}(k-2)\D0^n\D4\D1^{s}, \D1^{n-1-s}\D3\D3)$.\\
		After reading $x$, the NFA $M$ is in state $(B,j+1,1)$,
		where it can read neither \D1 nor \D3. Thus~$xv$ is not in $L$.
	\end{itemize}
\end{itemize}
We have shown (F2), which means that the set $\mathcal{F}$ 
is a fooling set for the language~$L$. 
Consider one more pair $(\D4^{m-1}\D3^{n-2}(k-2)\D0\D1,\varepsilon)$. 
The NFA $M$ may only be in the accepting state $q_{acc}$ after reading 
the string $\D4^{m-1}\D3^{n-2}(k-2)\D0\D1$. 
Since in this state no transitions are defined,
and the second part of each pair in $\mathcal{F}$ is nonempty,
the set 
$$\mathcal{F} \cup \{(\D4^{m-1}\D3^{n-2}(k-2)\D0\D1,\varepsilon)\}$$
is a fooling set for the language $L$ of size $2mn+2m+2n+1$.
This means that every  NFA for the language $L$ 
requires at least $2mn+2m+2n+1$ states.
\end{proof}

The above lower bound
is not applicable.
in the case of a pair of one-state automata.
In fact, in this special case
the complexity of this operation is lower.
While Lemma~\ref{le:upper_bound}
gives an upper bound of 7 states for this case,
6 states are actually sufficient.

\begin{lemma}\label{le:m=n=1}
Let $A$ and $B$ be two $1$-state NFAs
over an alphabet $\Sigma_k$.
Then the language $L(A) \plus L(B)$
is representable by an NFA with $6$ states.
\end{lemma}

\begin{proof}
Note that these 1-state NFAs must be partial DFAs.
Following the notation of Lemma~\ref{nfa_lower_bound_lemma},
let $0$ denote the state in the NFA $A$,
as well as the state in the NFA $B$.
If NFA $A$ has no transition on $k-1$,
then  state $(A, 0, 1)$ cannot be reached;
similarly for NFA $B$ and  state $(B, 0, 1)$.
If both $A$ and $B$ have transitions by $k-1$,
then states $(A, 0, 1)$ and $(B, 0, 1)$
can be merged into a state $q_{01}$,
which goes by $0$ to itself,
by a symbol $a+1$ to state $(A, 0, 0)$
if the NFA $A$ has a transition by $a$,
by a symbol $b+1$ to state $(B, 0, 0)$
if the NFA  $B$ has a transition by $b$,
for all $a,b$ in $\Sigma_k\setminus \{k-1\}$. 
\end{proof}


The next lemma establishes
a matching lower bound of 6 states.

\begin{lemma}
Let $\Sigma_k=\{\D0, \D1, \ldots, k-1\}$ be an alphabet with $k \geqslant 9$,
and consider $1$-state  partial DFAs $A$ and $B$ over $\Sigma_k$
which accept  languages $\{\D2,k-1\}^*$ and $\{\D3,k-1\}^*$, respectively.
Then every NFA
for $L(A) \plus L(B)$
has at least $6$ states.
\end{lemma}

\begin{proof}
Let $L=L(A) \plus L(B)$. 
Let the state in the NFA $A$ as well as the state in the NFA $B$
be denoted by 0.
Consider a six-state NFA for the language $L$
defined in Lemma~\ref{le:m=n=1},
with the states $(0,0,0)$, $(0,0,1)$, $q_{01}$,
$(A,0,0)$, $(B,0,0)$ and $q_{acc}$.
The transitions of this automaton
are shown in Figure~\ref{f:nfa_lower_bound_6}.
\begin{figure}[hbt]
	\centerline{\includegraphics[scale=0.85]{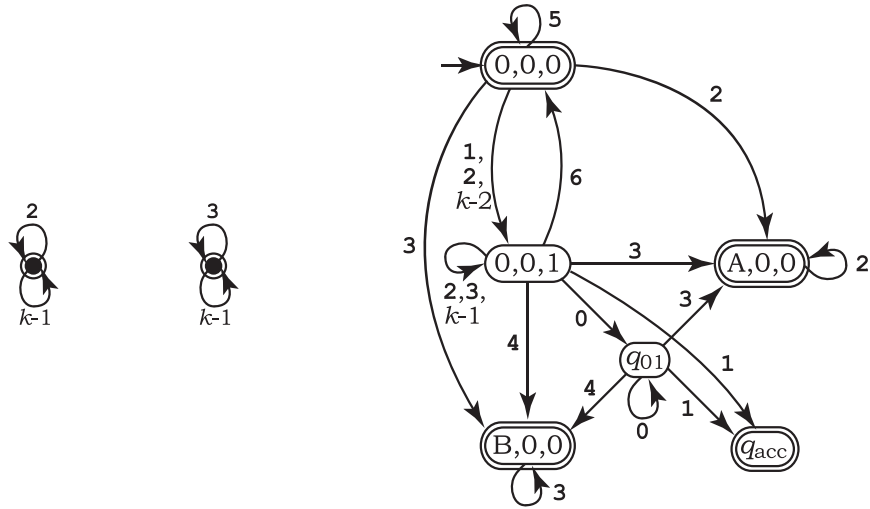}}
	\caption{The 1-state NFAs $A$ and $B$,
	and the 6-state NFA for $L(A) \plus L(B)$.}
	\label{f:nfa_lower_bound_6}
\end{figure}
Let
\begin{equation*}
\mathcal{A}=\{(\varepsilon, \D5),
	(k-2, \D6),
	((k-2)\D0, \D3\D2),
	((k-2)\D0\D3, \D2),
	((k-2)\D0\D4, \D3),
	((k-2)\D0\D1, \varepsilon) \},
\end{equation*}
and let us show that this set
is a fooling set for the language $L$.
Since the strings 
$\D5$, $(k-2)\D6$,
$(k-2)\D0\D3\D2$,
$(k-2)\D0\D4\D3$,
and $(k-2)\D0\D1$
are accepted by the NFA, the statement (F1) holds for $\mathcal{A}$.
On the other hand, the following strings are not accepted by this NFA:
the string $\D6$, any string starting with $(k-2)\D0$ and ending with $\D5$ or with $\D6$,
the strings $(k-2)\D0\D3\D3$, $(k-2)\D0\D4\D3\D2$, $(k-2)\D0\D4\D2$, 
and any string $(k-2)\D01w$ with $w\neq \varepsilon$.
This means that  the statement (F2) also holds for $\mathcal{A}$.
Hence $\mathcal{A}$ is a fooling set for the language $L$,
and so every NFA for this language needs at least 6 states.
\end{proof}

Putting together all the above lemmata,
the following result is obtained.

\begin{theorem}
For every $k \geqslant 9$,
the nondeterministic state complexity of positional addition
is given by the function
\begin{equation*}
	f_k(m,n)=\begin{cases}
		6,
			& \text{if\ } m=n=1, \\
		2mn+2m+2n+1,
			& \text{if\ } m+n \geqslant 3.
	\end{cases}
\end{equation*}
\end{theorem}


An obvious question left open in this paper
is the state complexity of positional addition
with respect to deterministic finite automata.
A straightforward upper bound is given by $2^{2mn+2m+2n+1}$,
though calculations show that
for small values of $k,m,n$
this bound is not reached.
Though the exact values of this complexity function
might involve too difficult combinatorics,
determining its asymptotics is an interesting problem,
which is proposed for future study.

\bibliographystyle{eptcs}
\bibliography{jiraskova}

\end{document}